\documentclass[a4paper,12pt]{article}
\usepackage[margin=2.6cm]{geometry}
\usepackage{PRIMEarxiv}
\AtBeginDocument{\setlength{\headheight}{15pt}}
\usepackage{amsmath,amssymb,amsthm}
\usepackage{enumitem}
\usepackage{booktabs}
\usepackage{float}
\usepackage{url}
\usepackage[T1]{fontenc}
\usepackage[utf8]{inputenc}
\usepackage{hyphenat}
\usepackage{microtype}
\usepackage{doi}
\usepackage{orcidlink}
\usepackage{tikz}
\usepackage{pgfplots}
\pgfplotsset{compat=1.18}
\usetikzlibrary{arrows.meta,positioning,fit,calc,backgrounds}
\definecolor{selred}{RGB}{200,30,30}
\definecolor{softgray}{gray}{0.82}

\newtheorem{theorem}{Theorem}
\newtheorem{lemma}{Lemma}
\newtheorem{proposition}{Proposition}

\theoremstyle{definition}

\newtheorem{example}{Example}
\theoremstyle{remark}

\newcommand{\ones}{\mathbf{1}}
\newcommand{\meet}{\wedge}
\newcommand{\join}{\vee}
\newcommand{\AuthorORCID}{https://orcid.org/0000-0002-7342-2090} 
\newcommand{\orcidifavailable}{\if\relax\detokenize\expandafter{\AuthorORCID}\relax\else\,\orcidlink{\AuthorORCID}\fi}
\title{Polynomial-Time  Riesz-Energy Subset Selection for Ordered Point Sets on Lines and $\ell_1$-Staircases}
\author{Michael T. M. Emmerich\orcidifavailable\\Faculty of Information Technology, University of Jyv\"askyl\"a, Finland}
\date{\today}

\begin{document}
\maketitle

\begin{abstract}
We study efficient algorithms for one-dimensional fixed-cardinality minimum Riesz $s$-energy subset selection on ordered real-line point sets and propose and test a polynomial-time exact s-t cut–based algorithm for this problem. Given $x_1<\cdots<x_n$, an exponent $s>0$, and a cardinality $k$, the task is to choose $1\leq i_1<\cdots<i_k\leq n$ minimizing
$E_s(i_1,\ldots,i_k)=\sum_{1\leq p<q\leq k}(x_{i_q}-x_{i_p})^{-s}$.
We prove that the one-dimensional Riesz interaction satisfies a Monge inequality. When feasible subsets are encoded as increasing index vectors, this property implies submodularity on a finite distributive lattice and yields polynomial-time solvability by submodular minimization over such lattices. The structural reduction holds for every real $s>0$. We also derive an explicit minimum $S$--$T$ cut formulation with $k(n-k)$ threshold variables and $O(k^2(n-k)^2)$ finite pairwise edges. The constructed graph has $N=k(n-k)$ nodes and $M=O(k^2(n-k)^2)$ arcs after an $O(k^2(n-k)^2)$ coefficient-construction step; an $O(NM)$ max-flow bound gives an $O(k^3(n-k)^3)$ cut step, while the conservative $O(N^2M)$ bound gives $O(k^4(n-k)^4)$. By an isometry argument, the same algorithm applies to $\ell_1$-staircases, including monotone two-dimensional Pareto-front and skyline approximations. The accompanying Python implementation includes verification examples and an empirical runtime benchmark; on balanced instances $n=2k$, the reference min-cut code overtakes exhaustive enumeration around $n=24$--$26$. The appendix provides examples and detailed explanations of the underlying theory.
\end{abstract}
\keywords{Riesz energy \and fixed-cardinality subset selection \and polynomial-time algorithms \and minimum cut \and submodular optimization \and ordered point sets \and multiobjective optimization \and  $\ell_1$-staircases}
\section{Introduction}

The Riesz $s$-energy of a finite point configuration is a classical repulsive
potential: pairs of points contribute more when they are close and less when
they are far apart. It has many applications, ranging from understanding particle distributions in physics
to guiding the search for evenly distributed approximation sets in set-oriented optimization 
and the design of experiments. 
In the fixed-cardinality subset-selection problem studied
here, one is given ordered points $x_1<x_2<\cdots<x_n$ on the real line and
must choose exactly $k$ of them so that the total interaction energy is as small
as possible.  Thus, for an increasing index vector $I=(i_1,\ldots,i_k)$ with
$1\leq i_1<\cdots<i_k\leq n$, the objective is
$E_s(I)=\sum_{1\leq p<q\leq k}(x_{i_q}-x_{i_p})^{-s}$, where $s>0$.  For
$k=0$, $k=1$, and $k=n$ the problem is trivial; throughout the genuinely
nontrivial part of this report we assume $2\leq k<n$.

\begin{example}[Ten equally spaced points]
Consider the ordered points $x_i=i-1$, $i=1,
\ldots,10$, with $k=4$ and $s=1$.  The optimum subset is
$S^*=\{x_1,x_4,x_7,x_{10}\}$, corresponding to the four positions
$0,3,6,9$.  The display in Figure \ref{fig:rieszex} shows the same solution in two ways: as a
subset of the ordered line and as the complete interaction graph of the selected
points, where edge widths are proportional to the Riesz weights.

\begin{figure}
    \caption{Illustration of the Riesz $s$ Energy subset selection problem for $k=4$ and $n=10$.}
    \label{fig:rieszex}
\begin{center}
\begin{tikzpicture}[font=\scriptsize,>=Latex]
\tikzset{
  introSel/.style={circle, draw=selred, fill=selred, inner sep=0pt, minimum size=5.4pt},
  introUnsel/.style={circle, draw=gray!70, fill=gray!20, inner sep=0pt, minimum size=5.4pt},
  introGuide/.style={gray!50, dashed, thin},
  wone/.style={line width=1.9pt},
  wtwo/.style={line width=1.15pt},
  wthree/.style={line width=0.55pt}
}

\begin{scope}[x=0.47cm,y=0.47cm]
\node[font=\bfseries\small] at (4.5,5.7) {A. Ordered points on the line};
\draw[thick,-{Latex[length=2mm]}] (-0.6,3.6) -- (9.8,3.6) node[right] {$\mathbb{R}$};
\foreach \i/\x in {1/0,2/1,3/2,4/3,5/4,6/5,7/6,8/7,9/8,10/9} {
  \draw[introGuide] (\x,3.6) -- (\x,2.55);
  \node at (\x,2.25) {\x};
}
\foreach \i/\x in {1/0,2/1,3/2,4/3,5/4,6/5,7/6,8/7,9/8,10/9} {
  \ifnum\i=1
    \node[introSel] at (\x,3.6) {};
    \node[text=selred] at (\x,4.85) {1};
  \else\ifnum\i=4
    \node[introSel] at (\x,3.6) {};
    \node[text=selred] at (\x,4.85) {4};
  \else\ifnum\i=7
    \node[introSel] at (\x,3.6) {};
    \node[text=selred] at (\x,4.85) {7};
  \else\ifnum\i=10
    \node[introSel] at (\x,3.6) {};
    \node[text=selred] at (\x,4.85) {10};
  \else
    \node[introUnsel] at (\x,3.6) {};
  \fi\fi\fi\fi
  \node at (\x,4.25) {$x_{\i}$};
}
\node[align=center] at (4.5,1.15) {Example: $n=10$, $k=4$, $s=1$\\[1mm]
Chosen optimum subset $S^*=\{\textcolor{selred}{x_1},\textcolor{selred}{x_4},\textcolor{selred}{x_7},\textcolor{selred}{x_{10}}\}$};
\end{scope}

\begin{scope}[shift={(7.7,0.00)},x=0.53cm,y=0.53cm]
\node[font=\bfseries\small] at (0,6.70) {B. Interaction graph of $S^*$};
\coordinate (T) at (0,5.10);
\coordinate (R) at (2.55,2.55);
\coordinate (B) at (0,0.00);
\coordinate (L) at (-2.55,2.55);

\draw[wthree] (L) -- (T) node[midway, above left=1pt] {$1/9$};
\draw[wone]   (T) -- (R) node[midway, above right=1pt] {$1/3$};
\draw[wtwo]   (T) -- (B) node[midway, above right= 1pt] {$1/6$};
\draw[wtwo]   (L) -- (R) node[midway, below =1pt] {$1/6$};
\draw[wone]   (L) -- (B) node[midway, below left=1pt] {$1/3$};
\draw[wone]   (B) -- (R) node[midway, below right=1pt] {$1/3$};

\node[introSel] at (T) {};
\node[introSel] at (R) {};
\node[introSel] at (B) {};
\node[introSel] at (L) {};
\node[above=2pt] at (T) {$x_1\,(0)$};
\node[right=4pt] at (R) {$x_4\,(3)$};
\node[below=2pt] at (B) {$x_7\,(6)$};
\node[left=4pt] at (L) {$x_{10}\,(9)$};
\end{scope}
\end{tikzpicture}
\end{center}
\begin{equation*}
E_1(S^*)
 = \sum_{i<j\in S^*}|x_j-x_i|^{-1}
 = \frac13+\frac16+\frac19+\frac13+\frac16+\frac13
 = \frac{13}{9}.
\end{equation*}
\end{figure}
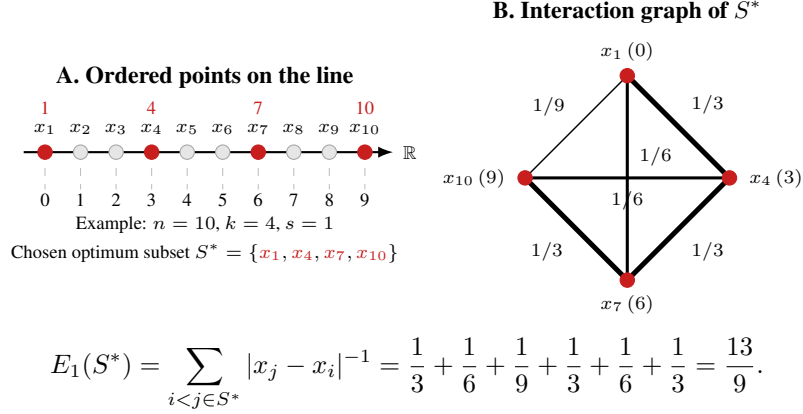
\end{example}

More formally, the input to the one-dimensional fixed-cardinality minimum
Riesz $s$-energy subset problem consists of rationally encoded real numbers
$x_1<x_2<\cdots<x_n$, a cardinality parameter $k\in\{0,\ldots,n\}$, and a
positive parameter $s$.  A feasible solution is a $k$-element subset of the
input points, equivalently an increasing index vector $I=(i_1,\ldots,i_k)$.
The task is to find a feasible vector minimizing $E_s(I)$.  The structural
arguments in this paper hold for every real $s>0$.  When bit-complexity is
considered, the coordinates are assumed to be rationally encoded and the
arithmetic assumptions on $s$ are those stated in the complexity section.

The guiding research question is whether this one-dimensional problem admits an
exact polynomial-time algorithm for arbitrary $s>0$ in the corresponding
arithmetic model.  Equivalently, we ask whether the open one-dimensional case
can be reduced to a polynomial-size combinatorial optimization problem whose
global optimum is precisely a minimum-energy subset.  In addition to the
proofs, the paper is accompanied by a working Python implementation of the
resulting cut algorithm, including the numerical examples and brute-force checks
used for verification.  The same implementation is used below for an empirical
runtime comparison against complete enumeration on balanced instances $n=2k$.

The paper first places the problem in the context of Riesz energy, diversity
indicators, and known hardness results.  It then shifts from ordinary subsets
to increasing index vectors and studies the distributive lattice induced by
componentwise order.  The central structural step is a Monge inequality for the
one-dimensional Riesz interaction, which yields submodularity of the objective
on this lattice.  This structure is then made algorithmic by introducing
threshold variables and constructing an explicit minimum $S$--$T$ cut instance.
The subsequent sections give the algorithm, prove its correctness, analyze its
complexity, and include small verification and runtime-efficiency examples.  The
paper concludes by relating the construction to general submodular minimization
and by discussing extensions and applications.

\section{Related work}

Riesz $s$-energy is a classical concept from potential theory; its
minimization leads to uniformly, or near-uniformly, dispersed particle
configurations on potentially nonlinear manifolds \cite{HardinSaff2005}.
It has recently also been studied as a diversity indicator in evolutionary
multiobjective optimization.  Falc\'on-Cardona, Uribe, and Rosas explicitly
proposed Riesz $s$-energy as such a diversity indicator and established the
corresponding supermodularity property of the set function~\cite{FalconCardona2024}.
Pereverdieva et al.\ subsequently discuss theoretical, computational, and
practical properties of Riesz $s$-energy and related diversity indicators,
including monotonicity-type properties, submodularity/supermodularity aspects,
and the NP-hardness of the corresponding subset-selection problem in general
metric spaces~\cite{Pereverdieva2024}.

The algorithmic route used in the present paper belongs to the literature on
submodular function minimization.  For ordinary set functions, strongly
polynomial-time algorithms were obtained independently by Schrijver and by
Iwata, Fleischer, and Fujishige~\cite{SchrijverSFM2000,IwataFleischerFujishige2001};
Orlin later gave a faster strongly polynomial-time algorithm~\cite{Orlin2009}.
For functions defined on a finite distributive lattice, Birkhoff's
representation theorem identifies the lattice with the ideals of its poset of
join-irreducibles, reducing submodular minimization on the lattice to
submodular minimization over a ring family of sets. Kolmogorov and Zabih \cite{KolmogorovZabih2004} established the link between S-T cuts and the optimization of submodular energy functions in computer graphics. A close algorithmic
precursor for the present route is also the work of de Berg, L\'opez Mart\'inez,
and Spieksma on finding diverse minimum $s$--$t$ cuts~\cite{DeBergDiverseCuts2024}.  They show that the pairwise-sum and coverage versions of $k$-diverse minimum
$s$--$t$ cuts can be solved in strongly polynomial time via submodular
minimization on a distributive lattice of ordered collections of minimum
$s$--$t$ cuts.  The later distributive-lattice
framework of the same authors abstracts this idea to broader classes of
combinatorial problems whose feasible solutions form a distributive lattice~\cite{DeBerg2025}.
The link in the present manuscript from the lattice/submodular function minimization viewpoint to a concrete
minimum-cut direction is due to this line of work (\cite{SD2026,DeBergDiverseCuts2024}).

The computational complexity of Riesz $s$-energy subset selection in
geometric settings was studied further by Emmerich, Pereverdieva, and
Deutz~\cite{Emmerich2026}, including a comparison with minimum pairwise
distance subset selection.  They showed that the problem is already
NP-hard in the Euclidean plane, with fixed dimension $2$, when $s$ is part
of the input.  The one-dimensional case, however, was left open.  An
earlier dynamic-programming approach for ordered point sets investigated
left--right recurrences for this setting~\cite{Emmerich2025DP}; this line
of work showed that the straightforward left--right dynamic-programming
strategy is not exact in the one-dimensional case.  The present note
addresses the remaining one-dimensional complexity question by exploiting
the lattice structure of increasing index vectors.

\section{The lattice of feasible index vectors}

Encode a feasible subset by its increasing index vector
\[
        I=(i_1,\ldots,i_k),
        \qquad
        1\leq i_1<\cdots<i_k\leq n.
\]
The set of such vectors is partially ordered by componentwise comparison:
\[
        I\leq J
        \quad\Longleftrightarrow\quad
        i_r\leq j_r \quad \text{for all } r=1,\ldots,k.
\]
The meet and join are
\[
        (I\meet J)_r=\min(i_r,j_r),
        \qquad
        (I\join J)_r=\max(i_r,j_r).
\]
These operations preserve strict increase, hence the feasible vectors form a finite distributive lattice; see, for example, standard treatments of finite distributive lattices and order ideals in \cite{Birkhoff1937,DaveyPriestley2002}.

It is convenient to remove the strict inequalities by writing $i_r=r+y_r$ for $r=1,\ldots,k$.  Here $y_r=i_r-r$ is not a coordinate of an input point, but the offset of the $r$-th chosen index from its smallest possible value $r$.  For example, if $n=7$, $k=3$, and the selected index vector is $I=(1,4,6)$, then $y_1=1-1=0$, $y_2=4-2=2$, and $y_3=6-3=3$, so $y=(0,2,3)$; the selected input points are still $x_1,x_4,x_6$.  With this notation, the strict inequalities $1\leq i_1<\cdots<i_k\leq n$ become the weak inequalities $0\leq y_1\leq y_2\leq\cdots\leq y_k\leq m$, where $m=n-k$.

It is useful to think of this transformation as subtracting the minimum staircase $(1,2,\ldots,k)$ from every feasible index vector.  Figure~\ref{fig:staircase-y-vectors} illustrates this viewpoint for two example index vectors.

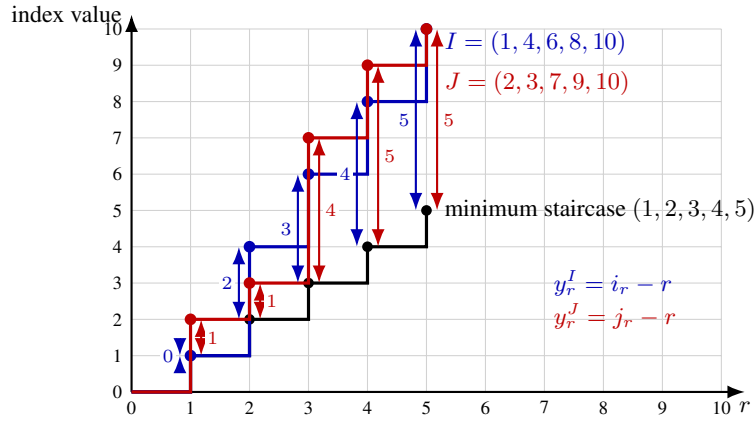
\begin{figure}[htbp]
\centering
\begin{tikzpicture}[x=0.78cm,y=0.48cm,>=Latex,font=\small]
  \draw[step=1,gray!35,thin] (0,0) grid (10,10);
  \draw[->,thick] (0,0) -- (10.4,0) node[below] {$r$};
  \draw[->,thick] (0,0) -- (0,10.4) node[left] {index value};
  \foreach \x in {0,...,10} {
    \node[below] at (\x,0) {\scriptsize \x};
    \node[left] at (0,\x) {\scriptsize \x};
  }

  \draw[very thick,black]
    (0,0)--(1,0)--(1,1)--(2,1)--(2,2)--(3,2)--(3,3)--(4,3)--(4,4)--(5,4)--(5,5);
  \foreach \r/\v in {1/1,2/2,3/3,4/4,5/5}
    \fill[black] (\r,\v) circle (2pt);
  \node[black,anchor=west] at (5.15,5) {minimum staircase $(1,2,3,4,5)$};

  \draw[very thick,blue!70!black]
    (0,0)--(1,0)--(1,1)--(2,1)--(2,4)--(3,4)--(3,6)--(4,6)--(4,8)--(5,8)--(5,10);
  \foreach \r/\v in {1/1,2/4,3/6,4/8,5/10}
    \fill[blue!70!black] (\r,\v) circle (2.2pt);
  \node[blue!70!black,anchor=west] at (5.15,9.6) {$I=(1,4,6,8,10)$};

  \draw[very thick,red!75!black]
    (0,0)--(1,0)--(1,2)--(2,2)--(2,3)--(3,3)--(3,7)--(4,7)--(4,9)--(5,9)--(5,10);
  \foreach \r/\v in {1/2,2/3,3/7,4/9,5/10}
    \fill[red!75!black] (\r,\v) circle (2.2pt);
  \node[red!75!black,anchor=west] at (5.15,8.5) {$J=(2,3,7,9,10)$};

  \foreach \r/\minv/\iv/\yv in {1/1/1/0,2/2/4/2,3/3/6/3,4/4/8/4,5/5/10/5} {
    \draw[<->,blue!70!black,thick] (\r-0.18,\minv) -- (\r-0.18,\iv)
      node[midway,left=1pt,fill=white,inner sep=1pt] {\scriptsize $\yv$};
  }
  \foreach \r/\minv/\jv/\yv in {1/1/2/1,2/2/3/1,3/3/7/4,4/4/9/5,5/5/10/5} {
    \draw[<->,red!75!black,thick] (\r+0.18,\minv) -- (\r+0.18,\jv)
      node[midway,right=1pt,fill=white,inner sep=1pt] {\scriptsize $\yv$};
  }

  \node[blue!70!black,anchor=west] at (7.0,3.0) {$y^I_r=i_r-r$};
  \node[red!75!black,anchor=west] at (7.0,2.1) {$y^J_r=j_r-r$};
\end{tikzpicture}
\caption{Three staircases on the integer grid.  The black staircase is the minimum staircase $(1,2,3,4,5)$.  The blue staircase corresponds to the subset indices $I=(1,4,6,8,10)$ and the red staircase to $J=(2,3,7,9,10)$.  The vertical double arrows show the coordinatewise differences from the minimum staircase, giving $y^I=(0,2,3,4,5)$ and $y^J=(1,1,4,5,5)$, where $y^I_r=i_r-r$ and $y^J_r=j_r-r$.}
\label{fig:staircase-y-vectors}
\end{figure}
Thus the problem is equivalent to minimizing
\[
        F(y_1,\ldots,y_k)
        =\sum_{1\leq p<q\leq k}
        \frac{1}{(x_{q+y_q}-x_{p+y_p})^s}
\]
over all nondecreasing integer vectors $y\in\{0,\ldots,m\}^k$.

\section{The Monge inequality}
The algorithm rests on the following elementary Monge property of the Riesz $s$-Energy in the 1-D case.

\begin{lemma}[One-dimensional Riesz interaction is Monge]
Let $s>0$ and $h(t)=t^{-s}$.  If
\[
        A<B<C<D,
\]
then
\[
        h(D-B)+h(C-A)\leq h(D-A)+h(C-B).
\]
\end{lemma}

\begin{proof}
Write
\[
        a=B-A,\qquad b=C-B,\qquad c=D-C,
\]
with $a,b,c>0$.  The desired inequality becomes
\[
        h(b+c)+h(a+b)
        \leq
        h(a+b+c)+h(b).
\]
Equivalently,
\[
        h(b)-h(b+c)
        \geq
        h(a+b)-h(a+b+c).
\]
Since
\[
        h'(t)=-s t^{-s-1}<0,
        \qquad
        h''(t)=s(s+1)t^{-s-2}>0,
\]
we have that $-h'(t)$ is decreasing in $t$.  Therefore the integral drop of $h$ over an interval of length $c$ decreases as the interval moves to the right:
\[
        h(u)-h(u+c)
        =\int_u^{u+c} -h'(t)\,dt.
\]
Taking $u=b$ and $u=a+b$ proves the claim.
\end{proof}

\begin{proposition}[Submodularity on the fixed-cardinality lattice]
The function $F$ is submodular on the lattice of feasible vectors.  That is, for all feasible $y,z$,
\[
        F(y)+F(z)
        \geq
        F(y\meet z)+F(y\join z),
\]
where meet and join are taken componentwise.
\end{proposition}

\begin{proof}
First isolate one pair of selected ranks $p<q$.  The contribution of this pair
is written as
\[
        V_{pq}(a,b)
        =\frac{1}{(x_{q+b}-x_{p+a})^s},
        \qquad 0\leq a\leq b\leq m.
\]
Here $a$ and $b$ are possible offset values for $y_p$ and $y_q$, respectively.
Thus, if $a=y_p$ and $b=y_q$, the actual selected indices in this pair are
$p+y_p$ and $q+y_q$, and the pair contributes
$(x_{q+y_q}-x_{p+y_p})^{-s}=V_{pq}(y_p,y_q)$.  For example, if $p=2$, $q=5$,
$y_2=3$, and $y_5=4$, then the selected indices in this pair are $2+3=5$ and
$5+4=9$, so the corresponding pair contribution is
$V_{2,5}(3,4)=(x_9-x_5)^{-s}$.

The reason it is enough to prove submodularity for these pair terms is that the
Riesz objective is exactly a sum over pairs:
\[
        F(y)=\sum_{1\leq p<q\leq k} V_{pq}(y_p,y_q).
\]
There are no triple terms or other global terms.  Since sums of submodular
functions are submodular, it is sufficient to show that every summand
$V_{pq}(y_p,y_q)$ satisfies the lattice submodularity inequality with respect
to the two coordinates on which it depends.

Fix $p<q$ and take two feasible two-coordinate arguments $(a,b)$ and $(c,d)$,
where $a\leq b$ and $c\leq d$.  These are the restrictions of two feasible
vectors to coordinates $p$ and $q$.  If the two arguments are ordered in the
same componentwise direction, for instance $a\leq c$ and $b\leq d$, then their
meet and join are simply $(a,b)$ and $(c,d)$, so the submodularity inequality is
an equality.  The same applies when $c\leq a$ and $d\leq b$.

Thus the only nontrivial case is the crossing case: one argument has the
smaller first coordinate but the larger second coordinate.  After possibly
exchanging the two arguments, we may assume $a\leq c$ and $d\leq b$.  Together
with feasibility of $(c,d)$, this gives the ordered chain
\[
        a\leq c\leq d\leq b.
\]
The meet and join of the two arguments are then
\[
        (a,b)\meet(c,d)=(a,d),
        \qquad
        (a,b)\join(c,d)=(c,b).
\]
Therefore the desired two-coordinate submodularity inequality is
\[
        V_{pq}(a,b)+V_{pq}(c,d)
        \geq
        V_{pq}(a,d)+V_{pq}(c,b).
\]
Substituting the definition of $V_{pq}$, this becomes
\[
        (x_{q+b}-x_{p+a})^{-s}
        +(x_{q+d}-x_{p+c})^{-s}
        \geq
        (x_{q+d}-x_{p+a})^{-s}
        +(x_{q+b}-x_{p+c})^{-s}.
\]
Now set
\[
        A=x_{p+a},\qquad B=x_{p+c},\qquad
        C=x_{q+d},\qquad D=x_{q+b}.
\]
Because $p<q$, $a\leq c\leq d\leq b$, and the input points are strictly ordered,
these four points satisfy $A\leq B<C\leq D$; the strict middle inequality follows
from $p+c<q+d$.  If $A=B$ or $C=D$, the desired inequality follows by continuity
from an arbitrarily small perturbation of equal points.  In the strict case,
the Monge lemma applied to $A<B<C<D$ gives
\[
        h(D-B)+h(C-A)\leq h(D-A)+h(C-B),
        \qquad h(t)=t^{-s}.
\]
Rearranging this inequality is exactly the displayed inequality for
$V_{pq}$. Hence, each pairwise term is submodular, and therefore, the full sum
$F$ is submodular (the idea for this proof originated in an online discussion \cite{SD2026}, see acknowledgement).
\end{proof}

This already implies a polynomial-time algorithm through submodular minimization over finite distributive lattices, or equivalently over the associated ring family of order ideals \cite{Birkhoff1937,Topkis1998,Fujishige2005,Schrijver2003,SchrijverSFM2000,IwataFleischerFujishige2001,Orlin2009,DeBerg2025}.  The next section gives a more explicit min-cut construction for this special problem.  This move from a distributive-lattice formulation toward an $S$--$T$ cut formulation is motivated by de Berg, L\'opez Mart\'inez, and Spieksma's  study of diverse minimum $s$--$t$ cuts, where ordered collections of minimum cuts are connected to distributive lattices and the resulting objectives are handled through submodular minimization~\cite{DeBergDiverseCuts2024}.

\section{Threshold variables, triangular expansion, and an explicit min-cut formulation}

The following threshold encoding is the point where the distributive-lattice
argument becomes a concrete graph-cut construction.  For each rank
$r=1,\ldots,k$ and threshold $t=1,\ldots,m$, introduce a binary variable
\[
        z_{r,t}=\ones[y_r\geq t].
\]
For a feasible vector $y$, let
\[
        Z(y)=\{(r,t):1\leq r\leq k,\;1\leq t\leq y_r\}.
\]
Then $Z(y)$ is closed under the implications
\[
        (r,t+1)\in Z(y) \Rightarrow (r,t)\in Z(y),
        \qquad
        (r,t)\in Z(y) \Rightarrow (r+1,t)\in Z(y).
\]
Equivalently, $Z(y)$ is an order ideal of the finite poset generated by
\[
        (r,t)\preceq (r,t+1),
        \qquad
        (r+1,t)\preceq (r,t).
\]
Conversely, every order ideal satisfying these closure relations determines a
unique feasible vector by
\[
        y_r=\sum_{t=1}^m z_{r,t}.
\]
Thus the feasible vectors are exactly the order ideals of this poset, ordered
by inclusion.  This is the explicit Birkhoff-type representation of the
finite distributive lattice of nondecreasing vectors
$0\leq y_1\leq\cdots\leq y_k\leq m$ \cite{Birkhoff1937,DaveyPriestley2002}.

In an $S$--$T$ cut, the source side represents variables with value $1$.  This is the standard maximum-closure/minimum-cut mechanism \cite{Picard1976}.  The
closure constraints are imposed by large-capacity arcs
\[
        (r,t+1)\to (r,t)
        \quad\text{and}\quad
        (r,t)\to(r+1,t).
\]
Throughout the algorithm we use a finite large capacity $U$.  It is enough to
choose $U$ larger than the total capacity of all finite source/sink and
interaction arcs, so that no minimum cut can profitably violate a closure
constraint.

The next lemma is the triangular threshold expansion needed for the pairwise
Riesz terms.  It avoids extending the pairwise cost from its natural triangular
domain $0\leq a\leq b\leq m$ to a full rectangle.

\begin{lemma}[Triangular threshold expansion]
Let
\[
        f:\{(a,b):0\leq a\leq b\leq m\}\to\mathbb{R}
\]
be arbitrary.  For $1\leq t<u\leq m$, define
\[
\begin{split}
        \Delta_{t,u}={}&
        f(t,u)-f(t-1,u)-f(t,u-1)+f(t-1,u-1),
\end{split}
\]
and define the unary coefficients by
\[
        C=f(0,0),
        \qquad
        \beta_u=f(0,u)-f(0,u-1),
        \quad u=1,\ldots,m,
\]
\[
        \alpha_a=f(a,a)-f(a-1,a-1)-\beta_a
        -\sum_{t=1}^{a-1}\Delta_{t,a},
        \quad a=1,\ldots,m.
\]
Then, for every $0\leq a\leq b\leq m$,
\[
\begin{split}
        f(a,b)={}&C+
        \sum_{t=1}^{a}\alpha_t+
        \sum_{u=1}^{b}\beta_u+
        \sum_{\substack{1\leq t\leq a\\1\leq u\leq b\\t<u}}
        \Delta_{t,u}.
\end{split}
\]
Equivalently, with threshold variables
$z^a_t=\ones[a\geq t]$ and $z^b_u=\ones[b\geq u]$,
\[
        f(a,b)=C+
        \sum_{t=1}^{m}\alpha_t z^a_t+
        \sum_{u=1}^{m}\beta_u z^b_u+
        \sum_{1\leq t<u\leq m}\Delta_{t,u} z^a_t z^b_u.
\]
\end{lemma}

\begin{proof}
Let $R(a,b)$ denote the right-hand side of the displayed expansion.  For
$a=0$, the formula reduces to
\[
        R(0,b)=f(0,0)+\sum_{u=1}^b(f(0,u)-f(0,u-1))=f(0,b).
\]
For diagonal points, the definition of $\alpha_a$ gives
\[
        R(a,a)-R(a-1,a-1)
        =\alpha_a+\beta_a+\sum_{t=1}^{a-1}\Delta_{t,a}
        =f(a,a)-f(a-1,a-1),
\]
so $R(a,a)=f(a,a)$ by induction on $a$.  Finally, for $b>a$,
\[
        R(a,b)-R(a,b-1)=\beta_b+\sum_{t=1}^{a}\Delta_{t,b}.
\]
By the definitions of $\beta_b$ and $\Delta_{t,b}$, the right-hand side
telescopes to
\[
        f(a,b)-f(a,b-1).
\]
Since the diagonal values are already correct, induction on $b$ proves the
claim for all $0\leq a\leq b\leq m$.
\end{proof}

For a fixed pair $p<q$, apply the lemma to
\[
        f(a,b)=V_{pq}(a,b)=\frac{1}{(x_{q+b}-x_{p+a})^s},
        \qquad 0\leq a\leq b\leq m.
\]
This gives the exact expansion
\[
\begin{split}
        V_{pq}(y_p,y_q)
        ={}& C_{pq}
        +\sum_{t=1}^{m} \alpha^{pq}_t z_{p,t}
        +\sum_{u=1}^{m} \beta^{pq}_u z_{q,u}  \\
        &+\sum_{1\leq t<u\leq m}
        \Delta^{pq}_{t,u} z_{p,t}z_{q,u}.
\end{split}
\]
The mixed coefficients are
\[
\begin{split}
        \Delta^{pq}_{t,u}={}&
        V_{pq}(t,u)-V_{pq}(t-1,u) \\
        &-V_{pq}(t,u-1)+V_{pq}(t-1,u-1),
        \qquad 1\leq t<u\leq m.
\end{split}
\]
They are nonpositive.  Indeed, set
\[
        A=x_{p+t-1},\qquad B=x_{p+t},\qquad
        C=x_{q+u-1},\qquad D=x_{q+u}.
\]
Because $p<q$ and $t<u$, the corresponding indices satisfy
\[
        p+t-1 < p+t < q+u-1 < q+u,
\]
so the four input points satisfy $A<B<C<D$.  The Monge lemma gives
\[
        h(D-B)+h(C-A)-h(D-A)-h(C-B)\leq0,
\]
which is exactly $\Delta^{pq}_{t,u}\leq0$.

Write
\[
        \Delta^{pq}_{t,u}=-w^{pq}_{t,u},
        \qquad
        w^{pq}_{t,u}\geq0.
\]
Then
\[
        -w^{pq}_{t,u} z_{p,t}z_{q,u}
        = -w^{pq}_{t,u}z_{p,t}
        + w^{pq}_{t,u} z_{p,t}(1-z_{q,u}).
\]
The second term is represented by a directed arc
\[
        (p,t)\to(q,u)
\]
of capacity $w^{pq}_{t,u}$, because this arc contributes exactly when
$z_{p,t}=1$ and $z_{q,u}=0$.  The remaining unary term is absorbed into the
source/sink arcs.  Thus the min-cut graph is not a separate idea from the
lattice formulation: it is a polynomial-size graph representation of the same
order-ideal minimization problem for this pairwise Monge objective.  The use of
nonpositive quadratic threshold interactions is the standard graph-cut
representability condition for submodular quadratic pseudo-Boolean terms
\cite{KolmogorovZabih2004}.

Unaries are handled in the usual way.  A term $c z_v$ with $c\geq0$ is
represented by an arc $v\to T$ of capacity $c$.  A term $c z_v$ with $c<0$ is
written as
\[
        c z_v = c + (-c)(1-z_v),
\]
so one adds the constant $c$ and an arc $S\to v$ of capacity $-c$.

\section{Algorithm}

The nontrivial threshold case is $2\leq k<n$ and $x_1<\cdots<x_n$.  For
$k=0$, $k=1$, or $k=n$, the optimum is immediate.

\begin{enumerate}[label=\textbf{Step \arabic*.}]
\item Sort the input points increasingly and check that they are distinct.
\item Set $m=n-k$ and introduce one binary threshold node $(r,t)$ for every
      $r=1,\ldots,k$ and $t=1,\ldots,m$.
\item For every pair $1\leq p<q\leq k$, compute the triangular expansion
      coefficients $C_{pq}$, $\alpha^{pq}$, $\beta^{pq}$, and
      $\Delta^{pq}_{t,u}$ for $1\leq t<u\leq m$.
\item For every coefficient $\Delta^{pq}_{t,u}=-w^{pq}_{t,u}<0$, add the
      finite graph-cut edge
      \[
          (p,t)\to(q,u)
      \]
      with capacity $w^{pq}_{t,u}$, and add the corresponding unary
      contribution $-w^{pq}_{t,u}z_{p,t}$ to the unary coefficient of
      $(p,t)$.
\item Add all unary terms as source/sink arcs and record the accumulated
      constant term.
\item Let $B$ be the sum of all finite source/sink and interaction capacities,
      and set $U=B+1$ or any larger valid capacity in the chosen arithmetic
      model.  Add closure arcs of capacity $U$ enforcing
      \[
          z_{r,t+1}\leq z_{r,t},
          \qquad
          z_{r,t}\leq z_{r+1,t}.
      \]
\item Compute a minimum $S$--$T$ cut.
\item Recover
      \[
          y_r=\sum_{t=1}^{m} z_{r,t},
          \qquad
          i_r=r+y_r.
      \]
\end{enumerate}

\section{Correctness}

\begin{theorem}
The min-cut algorithm returns a globally optimal cardinality-$k$ subset minimizing the one-dimensional Riesz $s$-energy.
\end{theorem}

\begin{proof}
Consider an $S$--$T$ cut and let $z_v=1$ exactly for nodes $v$ on the source
side.  Since the capacity $U$ of each closure arc is larger than the total
capacity of all finite arcs, a minimum cut cannot cross a closure arc: there is
always at least one closed feasible source set whose capacity is smaller than
$U$.  Therefore every minimum cut satisfies
\[
        z_{r,t+1}\leq z_{r,t},
        \qquad
        z_{r,t}\leq z_{r+1,t}.
\]
By the order-ideal representation above, such closed source sets are in
one-to-one correspondence with feasible vectors
\[
        0\leq y_1\leq\cdots\leq y_k\leq m,
        \qquad
        y_r=\sum_{t=1}^m z_{r,t},
\]
and hence with cardinality-$k$ subsets via $i_r=r+y_r$.

For each pair $p<q$, the triangular threshold expansion represents the pairwise
cost
\[
        V_{pq}(y_p,y_q)
        =\frac{1}{(x_{q+y_q}-x_{p+y_p})^s}
\]
exactly on the feasible domain $0\leq y_p\leq y_q\leq m$.  The mixed
coefficients are nonpositive by the Monge inequality.  Hence every mixed term
has the form $-wxy$ with $w\geq0$ and is represented exactly by a directed
edge plus the corresponding unary correction.

All unary terms are represented exactly by source/sink arcs, up to constants
independent of the cut.  Summing over all rank pairs, the capacity of every
closed finite cut equals
\[
        E_s(i_1,\ldots,i_k)+\text{constant},
\]
where the constant is independent of the chosen subset.  Minimizing cut
capacity is therefore equivalent to minimizing the Riesz energy.
\end{proof}

\section{Complexity}

Set
\[
        m=n-k.
\]
The graph has one threshold node for each pair $(r,t)$ with
$r=1,\ldots,k$ and $t=1,\ldots,m$.  Thus the number of nonterminal
nodes is
\[
        N=km=k(n-k).
\]
Including the source and sink only adds two further nodes.

The closure constraints contribute
\[
        k(m-1)+(k-1)m
        = 2km-k-m
        = O(km)
\]
closure arcs.  The unary terms contribute at most $O(km)$
source/sink arcs.  The dominant contribution comes from the pairwise
Riesz interactions.  For each rank pair $p<q$ and each threshold pair
$t<u$, there is at most one finite interaction arc.  Hence the number
of such arcs is bounded by
\[
        \binom{k}{2}\binom{m}{2}
        = O(k^2m^2)
        = O(k^2(n-k)^2).
\]
Consequently the total number of arcs satisfies
\[
        M = O(k^2m^2)=O(k^2(n-k)^2)
\]
in the worst case.

There is only one minimum $S$--$T$ cut to compute.  Therefore the
algorithmic cost is the cost of one maximum-flow/minimum-cut computation
on a graph with
\[
        N=k(n-k)
        \quad\text{nodes and}\quad
        M=O(k^2(n-k)^2)
        \quad\text{arcs}.
\]
Using, for example, the standard worst-case bound $O(N^2M)$ for a
Dinic-type implementation gives the conservative estimate
\[
        O(N^2M)
        = O\!\left(k^4(n-k)^4\right).
\]
In the balanced case $k=\Theta(n)$, this becomes $O(n^8)$ under this
conservative bound.  If one uses a more advanced strongly polynomial
maximum-flow algorithm with a bound of order $O(NM)$, the corresponding
estimate is
\[
        O(NM)
        = O\!\left(k^3(n-k)^3\right),
\]
which is $O(n^6)$ in the balanced case.  In practice, specialized graph-cut
implementations are often substantially faster than these worst-case
bounds suggest.

The graph construction itself requires evaluating the finite differences
for all rank pairs and threshold pairs.  With direct tabulation of the
values $V_{pq}(a,b)$ and prefix-sum bookkeeping for the unary corrections,
this can be done in
\[
        O(k^2m^2)=O(k^2(n-k)^2)
\]
time and space, up to the cost of arithmetic operations.  A naive reference
implementation may spend more time computing the same coefficients, but this
is not inherent in the construction.

For fixed $k$, we have $m=\Theta(n)$, so the graph has $O(n)$ nodes and
$O(n^2)$ arcs.  The same conservative $O(N^2M)$ bound then gives
$O(n^4)$ for the min-cut step.  For $k=\Theta(n)$, the graph has
$O(n^2)$ nodes and $O(n^4)$ arcs.

In exact arithmetic, if the points are rational and $s$ is a fixed positive
integer, then all energy values and capacities are rational numbers with
polynomially bounded encoding length.  Thus the above graph construction,
combined with any polynomial-time exact max-flow algorithm, gives a standard
Turing-polynomial algorithm.  More generally, the same statement applies in any
arithmetic model in which the required values $(x_j-x_i)^{-s}$ and their finite
differences can be evaluated and compared with polynomially many bit
operations.  For arbitrary non-integer or real $s$, the combinatorial
construction remains valid for every fixed input, but the result should be read
as a real-arithmetic or oracle-arithmetic algorithm unless such an evaluation
model is specified explicitly.

\subsection{Empirical Runtime Efficiency}
\label{subsec:empirical-runtime-efficiency}

The worst-case bounds above are polynomial but conservative.  To assess the
practical break-even point, I compared the explicit threshold min-cut
implementation with complete enumeration on balanced instances $n=2k$.  The
benchmark uses the companion Python code for the min-cut construction, with a
standard preflow-push maximum-flow routine, and a Numba-compiled exhaustive
combination enumerator for the baseline.  Thus the enumeration baseline is quite
favorable to brute force, while the min-cut code is a direct reference
implementation rather than a highly tuned graph-cut solver.  The test points are
deterministic mildly irregular ordered points,
\[
        x_i=i+0.05\sin(1.7i)+0.01(i/n)^2,
        \qquad i=0,\ldots,n-1,
\]
with $s=1$.  All instances for which both methods were run produced identical
energies and selected subsets.

For the theoretical comparison, complete enumeration with direct energy
evaluation performs $\Theta(k^2\binom{2k}{k})$ pair contributions when
$n=2k$.  Comparing this expression with the balanced conservative cut bound
$k^8=O(n^8)$ gives the first crossing at $k=13$, that is $n=26$.  If the looser
scale $n^8/2$ is used instead, the corresponding crossing is $n=36$.  Counting
only the number of subsets, without the $k^2$ pair-evaluation factor, shifts the
crossings to $n=38$ and $n=48$, respectively.

\begin{figure}[H]
\centering
\begin{tikzpicture}
\begin{axis}[
    width=0.92\linewidth,
    height=6.2cm,
    xlabel={$n$ with $k=n/2$},
    ylabel={CPU time in seconds},
    ymode=log,
    ymin=5e-6,
    ymax=25,
    xmin=4,
    xmax=60,
    xtick={4,10,20,30,40,50,60},
    grid=both,
    minor grid style={gray!15},
    major grid style={gray!30},
    legend style={at={(0.03,0.97)},anchor=north west,draw=none,fill=white},
    tick label style={font=\scriptsize},
    label style={font=\small},
    legend cell align={left}
]
\addplot+[mark=*,thick] coordinates {
    (4,0.00000872) (6,0.00000968) (8,0.00000601) (10,0.00000932)
    (12,0.00002790) (14,0.00012219) (16,0.00046071)
    (18,0.00200911) (20,0.00885193) (22,0.04348486)
    (24,0.17560995) (26,0.81302590) (28,3.47821405) (30,15.20843832)
};
\addlegendentry{complete enumeration}
\addplot+[mark=square*,thick,selred] coordinates {
    (4,0.00410619) (6,0.00041078) (8,0.00088098) (10,0.00200022)
    (12,0.00302792) (14,0.00852541) (16,0.00802257)
    (18,0.01684359) (20,0.02325281) (22,0.04035847)
    (24,0.18270274) (26,0.07867415) (28,0.11064484)
    (30,0.17179419) (36,0.54977612) (40,0.62227608)
    (44,1.24089540) (48,1.66338642) (52,2.37802957)
    (56,3.73192504) (60,5.38646040)
};
\addlegendentry{threshold min-cut}
\end{axis}
\end{tikzpicture}
\caption{Empirical runtime comparison on balanced instances $n=2k$ for $s=1$.
Enumeration was run until $n=30$, after which the number of subsets becomes too
large for a useful brute-force baseline in this environment.  The first noisy
crossing appears at $n=22$, but the min-cut implementation becomes clearly
faster from $n=26$ onward.  Times are single-run CPU wall-clock measurements and
are intended as a reproducibility check rather than as a tuned implementation
benchmark.}
\label{fig:runtime-break-even}
\end{figure}
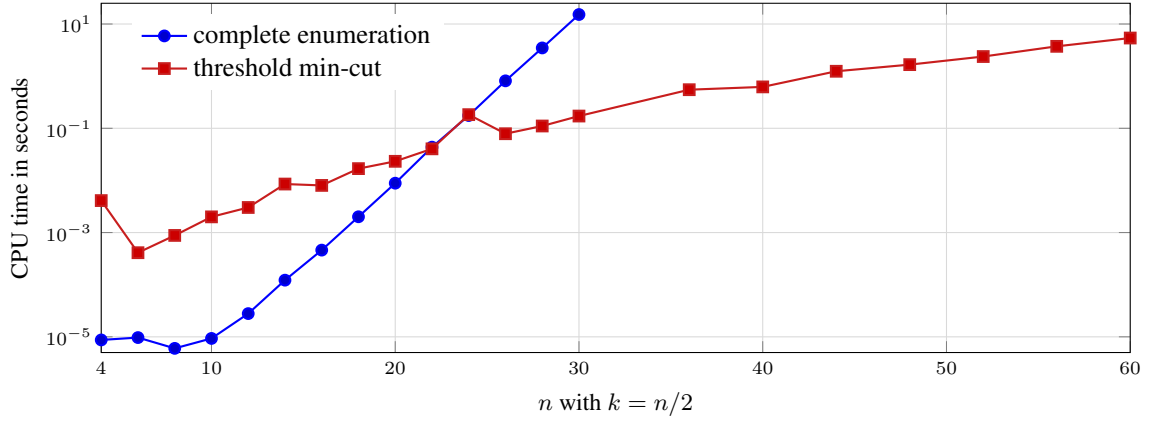

\begin{table}[H]
\centering
\small
\begin{tabular}{rrrrr}
\toprule
$n$ & $k$ & $\binom{n}{k}$ & enumeration (s) & min-cut (s) \\
\midrule
20 & 10 & 184,756 & 0.0089 & 0.0233 \\
22 & 11 & 705,432 & 0.0435 & 0.0404 \\
24 & 12 & 2,704,156 & 0.1756 & 0.1827 \\
26 & 13 & 10,400,600 & 0.8130 & 0.0787 \\
28 & 14 & 40,116,600 & 3.4782 & 0.1106 \\
30 & 15 & 155,117,520 & 15.2084 & 0.1718 \\
36 & 18 & 9,075,135,300 & -- & 0.5498 \\
48 & 24 & 32,247,603,683,100 & -- & 1.6634 \\
60 & 30 & 118,264,581,564,861,424 & -- & 5.3865 \\
\bottomrule
\end{tabular}
\caption{Representative timing data for the experiment in Figure~\ref{fig:runtime-break-even}.
For $n\geq 36$, complete enumeration was not run because the number of subsets
is already prohibitive.}
\label{tab:runtime-break-even}
\end{table}

The empirical behavior agrees with the simple theoretical comparison.  The
first measured crossing occurs at $n=22$, but $n=24$ is essentially a tie and is
slightly unfavorable to the cut implementation.  From $n=26$ onward, enumeration
is already more than one order of magnitude slower, and the gap widens rapidly.
This suggests that, despite the pessimistic $O(n^8)$ conservative bound, the
explicit min-cut formulation is computationally useful at sizes where complete
enumeration has only just become infeasible.

\section{Small verification examples}

For small instances, the min-cut algorithm can be checked against exhaustive enumeration.  The following representative instances were solved by the min-cut implementation and independently verified by brute-force enumeration.  The indices below are one-based, as in the mathematical formulation of this note.

\begin{itemize}[leftmargin=*]
\item $x=(0,1,2,3,4)$, $k=3$, $s=1$: optimum $(1,3,5)$ with $E_s=1.25$.
\item $x=(0,0.4,1.1,2.8,3,5)$, $k=3$, $s=2$: optimum $(1,4,6)$ with $E_s\approx0.3741626$.
\item $x=(0,1,10,11,12)$, $k=3$, $s=1$: optimum $(1,3,5)$ with $E_s\approx0.6833333$.
\item $x=(0,1,2,4,7,11)$, $k=4$, $s=1.5$: optimum $(1,4,5,6)$ with $E_s\approx0.5778501$.
\item $x=(0,1,2,3,10,11,12,20)$, $k=4$, $s=1$: optimum $(1,4,7,8)$ with $E_s\approx0.7616013$.
\item $x=(0,0.2,0.9,2.7,4.1,4.2,8)$, $k=3$, $s=3$: optimum $(1,5,7)$ with $E_s\approx0.0333205$.
\item $x=(0,5,6,7,8,20,21,40)$, $k=4$, $s=0.5$: optimum $(1,5,7,8)$ with $E_s\approx1.4134277$.
\item $x=(0,1,1.5,2.2,6,9,9.1,14,18)$, $k=5$, $s=2$: optimum $(1,5,7,8,9)$ with $E_s\approx0.2914438$.
\end{itemize}

In computational tests on random instances with $n\leq 10$, the min-cut solution agreed with brute-force enumeration for all tested values of $k$ and for several choices of $s$, including $s=0.5,1,1.5,2,3$.  The accompanying flat Python reference implementation contains the same examples and a randomized self-test routine.

\section{Relation to submodular minimization}

The min-cut construction is a specialized algorithmic realization of the more conceptual lattice argument.  The feasible vectors form a finite distributive lattice under componentwise minimum and maximum.  The threshold encoding above gives the relevant Birkhoff-type representation explicitly: feasible vectors correspond to order ideals of the threshold poset, and meet and join correspond to intersection and union of these ideals \cite{Birkhoff1937,DaveyPriestley2002}.  Therefore submodular minimization over this lattice is equivalent to submodular minimization over the associated ring family of closed sets \cite{Topkis1998,Fujishige2005,Schrijver2003}.  In particular, by applying Birkhoff's representation theorem and any polynomial-time submodular minimization algorithm for set functions, a submodular objective on a distributive lattice given through its join-irreducibles can be minimized in time polynomial in the number of join-irreducibles and in the evaluation-oracle time; see also the explicit use of this reduction in the distributive-lattice diversity framework of de Berg, L\'opez Mart\'inez, and Spieksma~\cite{DeBerg2025}.

The more concrete minimum-cut direction in this manuscript was suggested by the earlier work of de Berg, L\'opez Mart\'inez, and Spieksma on diverse minimum $s$--$t$ cuts~\cite{DeBergDiverseCuts2024}.  In that work, ordered collections of minimum cuts are related to distributive lattices, and pairwise-sum and coverage diversity objectives are solved via submodular function minimization.  Here the same conceptual bridge is used in the opposite direction: once the Riesz objective is shown to be submodular on the lattice of threshold ideals, the triangular expansion exposes a direct graph-cut representation for this particular pairwise Monge objective.

Thus, even without exploiting the explicit pairwise graph-cut representation, the one-dimensional fixed-cardinality Riesz energy problem is polynomial-time solvable by general submodular minimization over distributive lattices \cite{SchrijverSFM2000,IwataFleischerFujishige2001,Orlin2009}.  The graph-cut formulation is more concrete and simpler to implement for this particular objective because the triangular expansion turns each pairwise Monge term into graph-representable nonpositive threshold interactions plus unary terms.

\section{\texorpdfstring{$\ell_1$}{l1}-staircase Pareto front example}
\label{sec:l1-staircase-example}

A related direction concerns $\ell_1$-staircases, as defined in
\cite{Emmerich2026b}.  On a monotone sorted set, the $\ell_1$ distance is
isometric to distance on a line after choosing a consistent orientation of the
coordinates: along such a chain, all coordinate differences have fixed signs,
so the $\ell_1$ distance between two points is the absolute difference of a
single cumulative coordinate.  Two-dimensional Pareto fronts, after orienting
objectives consistently, fall into this category.  Thus the one-dimensional
construction suggests a direct route for $\ell_1$-staircase instances whose
candidate points form such a monotone chain.

We illustrate this reduction on the Pareto-front point set used in the
counterexample for the approximate dynamic-programming recurrence in
\cite{Emmerich2025DP}.  The candidate points are
\[
\begin{aligned}
P_1&=(2,20),  &P_2&=(4,18),  &P_3&=(6,16),  &P_4&=(9,12),\\
P_5&=(11,8), &P_6&=(14,5),  &P_7&=(17,3).
\end{aligned}
\]
They are sorted by increasing $f_1$ and decreasing $f_2$.  Hence, for
$i<j$,
\[
        \|P_j-P_i\|_1
        =(f_1(P_j)-f_1(P_i))+(f_2(P_i)-f_2(P_j)).
\]
Equivalently, after orienting the second coordinate, the map
\[
        \tau(P_i)=f_1(P_i)-f_2(P_i)+18
\]
embeds the staircase isometrically into the line.  For the above points this
one-dimensional coordinate is
\[
        \tau=(0,4,8,15,21,27,32).
\]
Therefore Riesz-energy subset selection with $\ell_1$ distance on the
staircase is exactly the one-dimensional Riesz-energy subset-selection problem
on these $\tau$-values.  For $s=1$ and $k=5$, exact minimization gives
\[
        S^*=\{P_1,P_3,P_4,P_6,P_7\},
\]
with
\[
\begin{aligned}
E_1^{\ell_1}(S^*)
&=\sum_{P_i,P_j\in S^*,\,i<j}\|P_j-P_i\|_1^{-1} \\
&=\frac1{8}+\frac1{15}+\frac1{27}+\frac1{32}
  +\frac1{7}+\frac1{19}+\frac1{24}
  +\frac1{12}+\frac1{17}+\frac1{5} \\
&\approx 0.83927.
\end{aligned}
\]

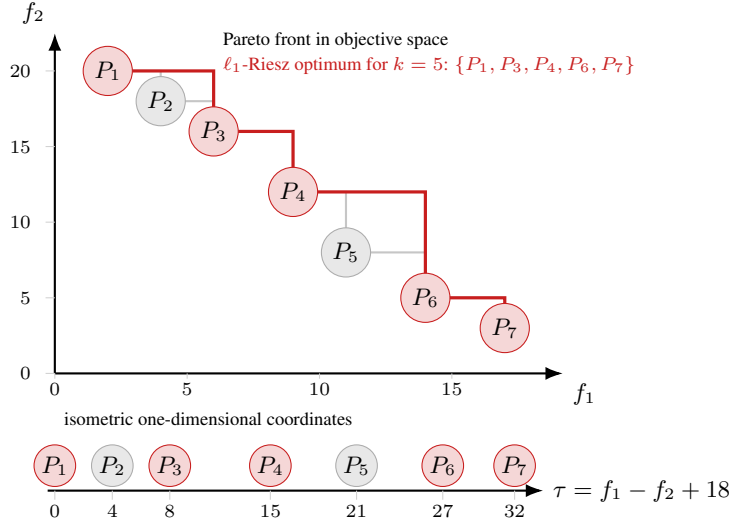
\begin{figure}[htbp]
\centering
\begin{tikzpicture}[
    >=Latex,
    every node/.style={font=\small},
    cand/.style={circle,draw=gray!65,fill=gray!18,inner sep=1.3pt,minimum size=6.5mm},
    opt/.style={circle,draw=selred,fill=selred!18,inner sep=1.3pt,minimum size=6.5mm},
    optline/.style={selred,very thick},
    stair/.style={gray!45,thick},
    taupt/.style={circle,draw=gray!65,fill=gray!18,inner sep=1.1pt,minimum size=5.5mm},
    tauopt/.style={circle,draw=selred,fill=selred!18,inner sep=1.1pt,minimum size=5.5mm}
]

\begin{scope}[x=0.35cm,y=0.20cm]
\draw[->,thick] (0,0) -- (19.2,0) node[below right] {$f_1$};
\draw[->,thick] (0,0) -- (0,22.5) node[above left] {$f_2$};

\foreach \x in {0,5,10,15} {
  \draw[gray!35] (\x,0) -- (\x,-0.35);
  \node[font=\scriptsize] at (\x,-1.0) {$\x$};
}
\foreach \y in {0,5,10,15,20} {
  \draw[gray!35] (0,\y) -- (-0.35,\y);
  \node[font=\scriptsize,anchor=east] at (-0.55,\y) {$\y$};
}

\draw[stair] (2,20) -- (4,20) -- (4,18)
             -- (6,18) -- (6,16)
             -- (9,16) -- (9,12)
             -- (11,12) -- (11,8)
             -- (14,8) -- (14,5)
             -- (17,5) -- (17,3);

\draw[optline] (2,20) -- (6,20) -- (6,16)
               -- (9,16) -- (9,12)
               -- (14,12) -- (14,5)
               -- (17,5) -- (17,3);

\node[opt]  (P1) at (2,20) {$P_1$};
\node[cand] (P2) at (4,18) {$P_2$};
\node[opt]  (P3) at (6,16) {$P_3$};
\node[opt]  (P4) at (9,12) {$P_4$};
\node[cand] (P5) at (11,8) {$P_5$};
\node[opt]  (P6) at (14,5) {$P_6$};
\node[opt]  (P7) at (17,3) {$P_7$};

\node[anchor=west,font=\scriptsize] at (6.0,22.0)
      {Pareto front in objective space};
\node[anchor=west,font=\scriptsize,text=selred] at (6.0,20.4)
      {$\ell_1$-Riesz optimum for $k=5$: $\{P_1,P_3,P_4,P_6,P_7\}$};
\end{scope}

\begin{scope}[shift={(0,-1.55)},x=0.19cm,y=1cm]
\draw[->,thick] (-0.5,0) -- (34,0) node[right] {$\tau=f_1-f_2+18$};
\foreach \t/\lab/\style in {0/P_1/tauopt,4/P_2/taupt,8/P_3/tauopt,15/P_4/tauopt,21/P_5/taupt,27/P_6/tauopt,32/P_7/tauopt} {
    \draw[gray!35] (\t,0) -- (\t,-0.12);
    \node[font=\scriptsize,below=2pt] at (\t,0) {$\t$};
    \node[\style] at (\t,0.33) {$\lab$};
}
\node[anchor=west,font=\scriptsize] at (0,0.95)
      {isometric one-dimensional coordinates};
\end{scope}

\end{tikzpicture}
\caption{A two-dimensional monotone Pareto front that is an $\ell_1$-staircase.
The candidate points are the same seven points used in the dynamic-programming
counterexample of \cite{Emmerich2025DP}.  Since $f_1$ increases and $f_2$
decreases along the sorted front, the $\ell_1$ distance between two points is
the difference of the cumulative coordinate $\tau=f_1-f_2+18$.  Thus the
$\ell_1$-Riesz subset-selection problem is isometric to the one-dimensional
problem on $\tau=(0,4,8,15,21,27,32)$.  For $s=1$ and $k=5$, exact minimization
selects $\{P_1,P_3,P_4,P_6,P_7\}$.}
\label{fig:l1-staircase-pareto-example}
\end{figure}

This example does not introduce a new problem class beyond the one-dimensional algorithm, but an application of the one-dimensional algorithm to a two-dimensional, or more generally n-dimensional, monotone staircase problem.  The monotone front can be measured either in
objective space with $\ell_1$ distance or on the line with the coordinate
$\tau$, and the Riesz energies agree exactly.  Consequently, the min-cut
algorithm developed for ordered one-dimensional point sets applies directly to
such $\ell_1$-staircase Pareto-front approximation problems.

\section{Conclusions and Outlook}

We have shown that the one-dimensional fixed-cardinality minimum Riesz
$s$-energy subset problem is polynomial-time solvable in the arithmetic models
specified above.  The structural reduction itself holds for every real $s>0$.
The proof rests on a
direct Monge property for the one-dimensional Riesz interaction and on the
distributive lattice structure of increasing index vectors.  Together, these
ingredients show that the objective is submodular on the feasible lattice and
therefore minimizable in polynomial time.

Beyond this conceptual reduction, the paper gives an explicit minimum
$S$--$T$ cut algorithm.  The construction uses $k(n-k)$ threshold variables
and $O(k^2(n-k)^2)$ finite pairwise edges.  Its preprocessing step has size
and time $O(k^2(n-k)^2)$, after which the remaining computational task is a
single maximum-flow/minimum-cut computation on a graph with
$N=k(n-k)$ nodes and $M=O(k^2(n-k)^2)$ arcs.  Under an $O(NM)$ max-flow
bound this gives an $O(k^3(n-k)^3)$ cut step, while the more conservative
$O(N^2M)$ bound gives $O(k^4(n-k)^4)$.  The accompanying open-source Python
implementation makes this cut construction directly reproducible and also
contains brute-force checks for small instances.  The empirical comparison in
Subsection~\ref{subsec:empirical-runtime-efficiency} shows that, for the
reference implementation on balanced test instances, the min-cut approach
overtakes a Numba-compiled exhaustive enumerator around $n=24$--$26$.

Section~\ref{sec:l1-staircase-example} illustrates the same construction for
monotone $\ell_1$-staircases, including a small Pareto-front approximation
example.  In such cases the $\ell_1$ distance is isometric to distance on a
line, so the one-dimensional min-cut algorithm applies directly after the
cumulative-coordinate transformation.

The present construction uses a direct graph-cut representation of all mixed threshold coefficients. Since these coefficients arise from a highly structured Riesz kernel rather than from an arbitrary submodular quadratic function, it is plausible that more compact exact representations are possible for special point sets, such as equally spaced points, or through interval/low-rank decompositions of the Monge coefficient matrices. We leave such refinements as future work.

Moreover, in applications such as Pareto-front and skyline approximation, append-only updates are not the most relevant dynamic case: new candidate points may appear at arbitrary positions in the sorted order, existing points may be removed, and small perturbations may change local spacings. A more useful dynamic viewpoint is to regard an insertion into the ordered input as the splitting of one adjacent gap, and a deletion as the merging of two adjacent gaps. In an implementation based on absolute threshold positions, $z_{r,j}=\mathbf{1}[i_r\ge j]$, an insertion at position $\ell$ introduces at most one new threshold node for each selected rank $r$. Since the finite pairwise graph-cut coefficients are discrete second differences involving two threshold gaps, all coefficients whose two gaps are unaffected can be copied, up to reindexing, while only the coefficients involving the split or merged gap have to be updated. This suggests local graph-edit work on the order of $O(k^2(n-k))$ per sorted insertion, deletion, or order-preserving local update, rather than reconstructing all $O(k^2(n-k)^2)$ finite interaction arcs. The previous residual network then provides a natural exact warm start for the next max-flow computation. Establishing a nontrivial worst-case bound for this dynamic reoptimization step remains open.

\appendix

\section{Appendix: A simple guide to order ideals, distributive lattices, and cuts}
\label{app:order-ideals}

This appendix explains the order-theoretic language used in the paper. It is purely expository and adds nothing new to the results. It is didactic and intended for readers without a specific background in these topics. The material is standard; see Birkhoff's representation theorem for finite distributive lattices \cite{Birkhoff1937} and modern introductions to order and lattice theory such as Davey and Priestley \cite{DaveyPriestley2002}. However, the purpose here is not to develop lattice theory in
general, but we would like to understand it in the context of the application of this paper, in which we use
 order ideals (or downward closures) and how a minimum cut is related to this representation.

\subsection{Partial orders and order ideals}

A partially ordered set, or poset, is a set $P$ together with a relation
$\preceq$ that is reflexive, antisymmetric, and transitive.  One should read
$x\preceq y$ as saying that $x$ is a prerequisite, lower element, or predecessor
of $y$.  In this paper the elements of the relevant poset are threshold nodes
$(r,t)$.

An \emph{order ideal}, also called a \emph{lower set} or \emph{down-set}, is a
subset $I\subseteq P$ with the following closure property:
\[
        y\in I \text{ and } x\preceq y
        \quad\Longrightarrow\quad
        x\in I.
\]

Thus, once an ideal contains an element, it must also contain all elements below
it.  In algorithmic terms, an order ideal is a set closed under prerequisite
constraints.

\subsection{A small order ideal of feasible vectors}

Before introducing threshold nodes, it is useful to look directly at a small
poset of feasible vectors.  Consider the feasible vectors of length three with
entries between $0$ and $2$,
\[
        \mathcal Y_3=\{(y_1,y_2,y_3):0\leq y_1\leq y_2\leq y_3\leq 2\}.
\]
We order these vectors componentwise:
\[
        y\leq y'
        \quad\Longleftrightarrow\quad
        y_1\leq y'_1,\quad y_2\leq y'_2,\quad y_3\leq y'_3.
\]
An example of an order ideal is the principal ideal generated by
$a=(0,1,2)$:
\[
        \downarrow a
        =\{y\in \mathcal Y_3:y\leq (0,1,2)\}
        =\{(0,0,0),(0,0,1),(0,0,2),(0,1,1),(0,1,2)\}.
\]
In the Hasse diagram below, these five vectors are shaded.  The downward-closed
property means that, once the ideal contains $(0,1,2)$, it must contain every
vector below it, such as $(0,0,2)$ and $(0,1,1)$.  It does not have to contain
$(1,1,2)$, because $(1,1,2)\nleq (0,1,2)$.

\begin{figure}[htbp]
\centering
\begin{tikzpicture}[
    scale=0.66,
    every node/.style={font=\small},
    vec/.style={draw,rounded corners,inner sep=2pt,minimum width=16mm},
    ideal/.style={draw,rounded corners,inner sep=2pt,minimum width=16mm,fill=blue!15}
]
\node[ideal] (000) at (0,0) {$(0,0,0)$};

\node[ideal] (001) at (0,1.05) {$(0,0,1)$};

\node[ideal] (002) at (-1.25,2.1) {$(0,0,2)$};
\node[ideal] (011) at ( 1.25,2.1) {$(0,1,1)$};

\node[ideal] (012) at (-0.65,3.15) {$(0,1,2)$};
\node[vec]   (111) at ( 1.85,3.15) {$(1,1,1)$};

\node[vec]   (022) at (-1.25,4.2) {$(0,2,2)$};
\node[vec]   (112) at ( 1.25,4.2) {$(1,1,2)$};

\node[vec]   (122) at (0,5.25) {$(1,2,2)$};

\node[vec]   (222) at (0,6.3) {$(2,2,2)$};

\draw (000)--(001);
\draw (001)--(002);
\draw (001)--(011);
\draw (002)--(012);
\draw (011)--(012);
\draw (011)--(111);
\draw (012)--(022);
\draw (012)--(112);
\draw (111)--(112);
\draw (022)--(122);
\draw (112)--(122);
\draw (122)--(222);
\end{tikzpicture}
\caption{The componentwise partial order on the feasible vectors
$0\leq y_1\leq y_2\leq y_3\leq 2$.  The shaded set is the order ideal
$\downarrow(0,1,2)$.}
\label{fig:y-vector-order-ideal}
\end{figure}
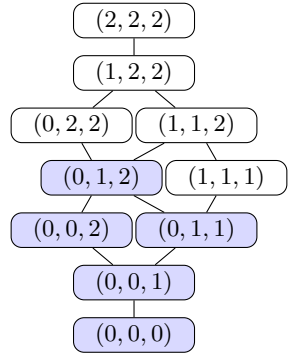

The threshold-node construction used in the main text is a more economical way
of representing such downward closure.  Instead of drawing the entire lattice of
feasible $y$-vectors, it records which threshold events $y_r\geq t$ have become
true.

For the threshold variables in this paper,
\[
        z_{r,t}=1 \quad\Longleftrightarrow\quad y_r\geq t,
\]
the two natural implications are
\[
        z_{r,t+1}=1 \Rightarrow z_{r,t}=1,
        \qquad
        z_{r,t}=1 \Rightarrow z_{r+1,t}=1.
\]
The first implication says that if $y_r$ reaches level $t+1$, then it also
reaches level $t$.  The second says that if $y_r\geq t$ and the vector is
nondecreasing, then $y_{r+1}\geq t$.  Equivalently, the selected threshold nodes
form an order ideal in the poset generated by
\[
        (r,t)\preceq (r,t+1),
        \qquad
        (r+1,t)\preceq (r,t).
\]
The direction of $\preceq$ is chosen so that membership of an upper threshold
forces membership of all lower/prerequisite thresholds.

\subsection{From feasible vectors to ideals and back}

Let
\[
        0\leq y_1\leq y_2\leq\cdots\leq y_k\leq m.
\]
Given such a vector, define
\[
        I(y)=\{(r,t):1\leq r\leq k,\;1\leq t\leq y_r\}.
\]
Then $I(y)$ is an order ideal.  Conversely, if $I$ is an order ideal in the
threshold poset, define
\[
        y_r=|\{t:(r,t)\in I\}|.
\]
The vertical closure implication implies that each column is an initial segment
of thresholds, so this formula recovers a well-defined height $y_r$.  The
horizontal implication implies $y_r\leq y_{r+1}$.  Therefore $I$ determines a
feasible vector.  These two constructions are inverse to each other:
\[
        y
        \quad\longleftrightarrow\quad
        I(y).
\]
This is the concrete version, for the present problem, of the general
Birkhoff representation principle that finite distributive lattices can be
represented by order ideals of a poset \cite{Birkhoff1937,DaveyPriestley2002}.

\subsection{Why this is a distributive lattice}

The feasible vectors are ordered componentwise:
\[
        y\leq y'
        \quad\Longleftrightarrow\quad
        y_r\leq y'_r \text{ for all } r.
\]
The meet and join are componentwise minimum and maximum:
\[
        (y\wedge y')_r=\min(y_r,y'_r),
        \qquad
        (y\vee y')_r=\max(y_r,y'_r).
\]
Under the ideal representation, these operations become set intersection and
set union:
\[
        I(y\wedge y')=I(y)\cap I(y'),
        \qquad
        I(y\vee y')=I(y)\cup I(y').
\]
Since union and intersection distribute over each other, the feasible vectors
form a finite distributive lattice.  This is useful because submodularity is
naturally stated on lattices: a function $F$ is submodular if
\[
        F(y)+F(y')\geq F(y\wedge y')+F(y\vee y')
\]
for all feasible $y,y'$.  This is the lattice analogue of the familiar
submodularity inequality for set functions.  General references for submodular
functions and submodularity on ordered structures include Topkis
\cite{Topkis1998}, Fujishige \cite{Fujishige2005}, and Schrijver
\cite{Schrijver2003}.

\subsection{Why submodularity gives a lattice reduction}

The Monge inequality proved in the main text shows that each pairwise Riesz
term is submodular with respect to the above meet and join.  Since sums of
submodular functions are submodular, the full Riesz objective is submodular on
the distributive lattice of feasible vectors.  Hence the original subset problem
is reduced to the following standard form: \textbf{minimize a submodular function over the order ideals of a finite poset.}
This is the conceptual lattice reduction.  It is independent of the later
special graph construction.  The graph construction is an additional benefit of
the particular pairwise Monge structure: after the triangular threshold
expansion, all mixed threshold coefficients are nonpositive, which makes the
objective graph-representable by one minimum cut.

\subsection{Why a minimum cut enforces order ideals}

Consider a directed graph with source $S$, sink $T$, and one node for each
threshold variable.  An $S$--$T$ cut partitions the nodes into a source side and
a sink side.  We interpret
\[
        z_{r,t}=1
        \quad\Longleftrightarrow\quad
        (r,t) \text{ lies on the source side.}
\]
To enforce an implication
\[
        z_a=1 \Rightarrow z_b=1,
\]
add a large-capacity arc $a\to b$.  If a cut places $a$ on the source side and
$b$ on the sink side, then this arc is cut and the cut pays the large capacity.
Choosing the large capacity $U$ larger than the total finite capacity guarantees
that no optimum cut violates the implication.  Therefore the source side of an
optimal finite cut is closed under all implications; it is exactly an order
ideal.

This is the classical maximum-closure/minimum-cut idea: optimizing over closed
sets of a directed graph can be reduced to an $S$--$T$ cut by using large
capacity arcs for closure constraints and source/sink arcs for node weights
\cite{Picard1976}.  In the present paper, the closed sets are precisely the
threshold ideals corresponding to feasible nondecreasing vectors.

\subsection{Two elementary cuts in a weighted directed graph}

Before returning to threshold ideals, it is useful to distinguish an arbitrary
$S$--$T$ cut from a minimum $S$--$T$ cut in a plain weighted directed graph.
Figure~\ref{fig:two-st-cuts} shows the same graph with two different cuts.  A
cut is evaluated by summing the capacities of the arcs that leave the source
side and enter the sink side.  The first cut is feasible but not optimal: it
pays $4+2+5=11$.  The second cut separates the source immediately and pays only
$3+2=5$, and is the minimum cut for this graph.  In the threshold construction
used in this paper, the graph is built so that cuts of the second kind are not
chosen merely because they are cheap: large-capacity closure arcs force the
source side to be an order ideal, and the remaining finite arcs encode the
objective value.

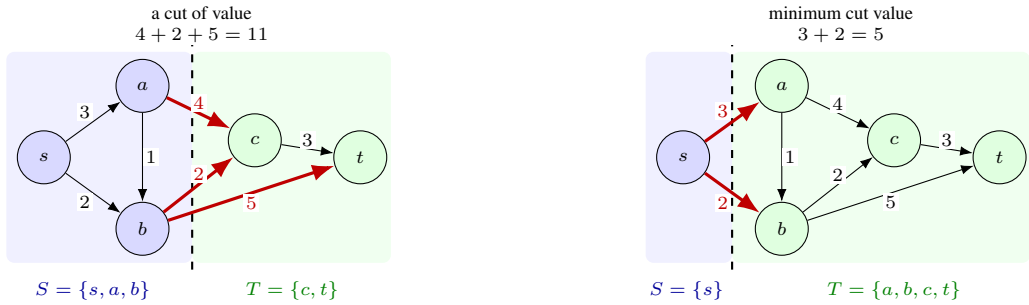
\begin{figure}[htbp]
\centering
\begin{minipage}[t]{0.49\textwidth}
\centering
\begin{tikzpicture}[
    >=Latex,
    scale=0.9,
    every node/.style={font=\scriptsize},
    v/.style={draw,circle,minimum size=7mm,inner sep=0pt,fill=white},
    sourcev/.style={draw,circle,minimum size=7mm,inner sep=0pt,fill=blue!15},
    sinkv/.style={draw,circle,minimum size=7mm,inner sep=0pt,fill=green!12},
    edge/.style={->,thin},
    cutedge/.style={->,very thick,red!75!black},
    lab/.style={fill=white,inner sep=1pt}
]
\coordinate (cut) at (2.18,0);
\node[sourcev] (s) at (0,0) {$s$};
\node[sourcev] (a) at (1.45,1.05) {$a$};
\node[sourcev] (b) at (1.45,-1.05) {$b$};
\node[sinkv] (c) at (3.10,0.25) {$c$};
\node[sinkv] (t) at (4.65,0) {$t$};
\begin{scope}[on background layer]
  \fill[blue!6,rounded corners=3pt] (-0.55,-1.55) rectangle (2.18,1.55);
  \fill[green!6,rounded corners=3pt] (2.18,-1.55) rectangle (5.10,1.55);
\end{scope}
\draw[dashed,thick] (2.18,-1.65) -- (2.18,1.65);
\draw[edge] (s) -- node[lab,above left] {$3$} (a);
\draw[edge] (s) -- node[lab,below left] {$2$} (b);
\draw[edge] (a) -- node[lab,right] {$1$} (b);
\draw[edge] (c) -- node[lab,above] {$3$} (t);
\draw[cutedge] (a) -- node[lab,above] {$4$} (c);
\draw[cutedge] (b) -- node[lab,above] {$2$} (c);
\draw[cutedge] (b) -- node[lab,below] {$5$} (t);
\node[blue!60!black] at (0.72,-1.95) {$S=\{s,a,b\}$};
\node[green!50!black] at (3.65,-1.95) {$T=\{c,t\}$};
\node[align=center] at (2.32,1.95) {a cut of value\\$4+2+5=11$};
\end{tikzpicture}
\end{minipage}\hfill
\begin{minipage}[t]{0.49\textwidth}
\centering
\begin{tikzpicture}[
    >=Latex,
    scale=0.9,
    every node/.style={font=\scriptsize},
    v/.style={draw,circle,minimum size=7mm,inner sep=0pt,fill=white},
    sourcev/.style={draw,circle,minimum size=7mm,inner sep=0pt,fill=blue!15},
    sinkv/.style={draw,circle,minimum size=7mm,inner sep=0pt,fill=green!12},
    edge/.style={->,thin},
    cutedge/.style={->,very thick,red!75!black},
    lab/.style={fill=white,inner sep=1pt}
]
\node[sourcev] (s) at (0,0) {$s$};
\node[sinkv] (a) at (1.45,1.05) {$a$};
\node[sinkv] (b) at (1.45,-1.05) {$b$};
\node[sinkv] (c) at (3.10,0.25) {$c$};
\node[sinkv] (t) at (4.65,0) {$t$};
\begin{scope}[on background layer]
  \fill[blue!6,rounded corners=3pt] (-0.55,-1.55) rectangle (0.72,1.55);
  \fill[green!6,rounded corners=3pt] (0.72,-1.55) rectangle (5.10,1.55);
\end{scope}
\draw[dashed,thick] (0.72,-1.65) -- (0.72,1.65);
\draw[cutedge] (s) -- node[lab,above left] {$3$} (a);
\draw[cutedge] (s) -- node[lab,below left] {$2$} (b);
\draw[edge] (a) -- node[lab,right] {$1$} (b);
\draw[edge] (a) -- node[lab,above] {$4$} (c);
\draw[edge] (b) -- node[lab,above] {$2$} (c);
\draw[edge] (b) -- node[lab,below] {$5$} (t);
\draw[edge] (c) -- node[lab,above] {$3$} (t);
\node[blue!60!black] at (0.05,-1.95) {$S=\{s\}$};
\node[green!50!black] at (3.10,-1.95) {$T=\{a,b,c,t\}$};
\node[align=center] at (2.32,1.95) {minimum cut value\\$3+2=5$};
\end{tikzpicture}
\end{minipage}
\caption{Two cuts of the same weighted directed graph.  Red arcs are precisely
those leaving the source side and entering the sink side, hence they contribute
to the cut value.  The left partition is only a cut, while the right partition
is a minimum cut.}
\label{fig:two-st-cuts}
\end{figure}

\subsection{Min-cut examples for \texorpdfstring{$k$}{k}-subset selection}

The previous figure showed how an $S$--$T$ cut is evaluated in an ordinary
weighted graph.  We now use the same cut viewpoint for the threshold graph that
appears in the Riesz-energy algorithm.  Take $k=2$ and $m=n-k=2$.  Then
there are four threshold variables,
\[
        z_{1,1},\ z_{1,2},\ z_{2,1},\ z_{2,2},
\]
corresponding to the four threshold nodes $(1,1),(1,2),(2,1),(2,2)$.  The
closure implications are
\[
        z_{1,2}\Rightarrow z_{1,1},\qquad
        z_{2,2}\Rightarrow z_{2,1},\qquad
        z_{1,1}\Rightarrow z_{2,1},\qquad
        z_{1,2}\Rightarrow z_{2,2}.
\]
Figure~\ref{fig:tiny-cut-example} shows the corresponding closure graph.  The
shaded source-side nodes represent the feasible vector $y=(1,2)$, because the
selected thresholds are exactly
\[
        I(y)=\{(1,1),(2,1),(2,2)\}.
\]
In the right panel, the source side of the cut is indicated explicitly by a
blue dashed stair-shaped region containing exactly $S$, $z_{1,1}$, $z_{2,1}$, and $z_{2,2}$;
all remaining nodes lie on the sink side.  The red arc is the only finite arc
crossing from the source side to the sink side in this illustrative cut.  In
particular, the arc $S\to z_{2,2}$ of capacity $2$ is not cut, because both of
its endpoints lie on the source side.  A directed cut counts only arcs from the
source side to the sink side.  Thus, among the finite arcs drawn in the right
panel, the displayed cut contributes only the capacity of $z_{2,1}\to T$.

The selected subset is recovered from the row counts of the selected threshold
nodes.  Since $I(y)$ contains one node in row $1$ and two nodes in row $2$, we
get $y=(1,2)$.  The original selected indices are $i_r=r+y_r$, hence
\[
        (i_1,i_2)=(1+1,2+2)=(2,4).
\]
Thus this cut encodes the subset $\{x_2,x_4\}$.  This tiny picture should be
read as an encoding example rather than as a full numerical Riesz instance: the
finite capacities shown are representative source/sink and pairwise arcs from
the graph-cut construction.

This set is an order ideal: whenever a node is selected, all of its
prerequisites are selected as well.  Equivalently, in the cut graph every
large-capacity closure arc starts and ends on the source side or the sink side,
so none of them crosses the cut.  By contrast, if one tried to put $z_{1,1}=1$
on the source side but $z_{2,1}=0$ on the sink side, then the large-capacity arc
$z_{1,1}\to z_{2,1}$ would be cut, making such a partition nonoptimal.

\begin{figure}[htbp]
\centering
\begin{minipage}[t]{0.42\textwidth}
\centering
\begin{tikzpicture}[
    >=Latex,
    node distance=11mm and 11mm,
    every node/.style={font=\small},
    pnode/.style={draw,circle,inner sep=1.5pt,minimum size=7mm},
    chosen/.style={draw,circle,inner sep=1.5pt,minimum size=7mm,fill=blue!15}
]
\node[chosen] (a) at (0,0) {$(1,1)$};
\node[pnode]  (b) at (0,1.5) {$(1,2)$};
\node[chosen] (c) at (2.4,0) {$(2,1)$};
\node[chosen] (d) at (2.4,1.5) {$(2,2)$};
\draw[->,thick] (b) -- (a);
\draw[->,thick] (d) -- (c);
\draw[->,thick] (a) -- (c);
\draw[->,thick] (b) -- (d);
\node[align=center] at (1.2,-1.0) {Threshold poset and ideal\\for $y=(1,2)$};
\end{tikzpicture}
\end{minipage}\hfill
\begin{minipage}[t]{0.54\textwidth}
\centering
\begin{tikzpicture}[
    >=Latex,
    node distance=9mm and 12mm,
    every node/.style={font=\small},
    var/.style={draw,circle,inner sep=1.5pt,minimum size=7mm},
    chosen/.style={draw,circle,inner sep=1.5pt,minimum size=7mm,fill=blue!15},
    term/.style={font=\scriptsize,fill=white,inner sep=1pt}
]
\node[draw,rectangle,rounded corners,minimum width=7mm,minimum height=7mm,fill=blue!10] (S) at (-1.6,0.75) {$S$};
\node[draw,rectangle,rounded corners,minimum width=7mm,minimum height=7mm] (T) at (5.2,0.75) {$T$};

\node[chosen] (v11) at (0.3,0) {$z_{1,1}$};
\node[var]    (v12) at (0.3,1.5) {$z_{1,2}$};
\node[chosen] (v21) at (2.6,0) {$z_{2,1}$};
\node[chosen] (v22) at (2.6,1.5) {$z_{2,2}$};

\draw[->] (S) -- node[term,pos=0.42,above left=-1pt and -1pt] {$1$} (v11);
\draw[->] (S) -- node[term,pos=0.58,above left=-1pt and -1pt] {$2$} (v22);
\draw[->] (v12) -- node[term,pos=0.62,above] {$1$} (T);
\draw[->,red!75!black,very thick] (v21) -- node[term,pos=0.48,below,text=red!75!black] {$1$} (T);

\draw[->,thick] (v12) -- node[term,left] {$U$} (v11);
\draw[->,thick] (v22) -- node[term,right] {$U$} (v21);
\draw[->,thick] (v11) -- node[term,below] {$U$} (v21);
\draw[->,thick] (v12) -- node[term,above] {$U$} (v22);

\draw[->,dashed] (v11) to[bend left=22] node[term,pos=0.53,above=3pt] {$w$} (v22);

\begin{scope}[on background layer]
    \path[draw=blue!60!black,dashed,very thick,rounded corners=8pt,fill=blue!6]
        (-2.20,1.25) -- (-1.05,1.25) -- (-1.05,1.05) --
        (1.95,1.05) -- (1.95,1.95) -- (3.15,1.95) --
        (3.15,-0.50) -- (-0.25,-0.50) -- (-0.25,0.35) --
        (-2.20,0.35) -- cycle;
\end{scope}
\node[blue!60!black,font=\scriptsize,anchor=south west] at (-2.17,1.34) {source side of cut};
\node[red!75!black,font=\scriptsize,anchor=west] at ($(v21)!0.54!(T)+(0.12,-0.50)$) {cut edge};
\draw[red!75!black,->] ($(v21)!0.57!(T)+(0.25,-0.34)$) -- ($(v21)!0.66!(T)$);
\node[align=center] at (1.9,-1.15) {Example cut: source side encodes\\$I(y)=\{(1,1),(2,1),(2,2)\}$, $y=(1,2)$\\and subset $\{x_2,x_4\}$};
\end{tikzpicture}
\end{minipage}
\caption{A tiny illustrative example with $k=2$ and $m=2$.  Left: the
threshold poset and the order ideal corresponding to $y=(1,2)$.  Right: the
associated cut graph.  The blue dashed stair-shaped region is the source side of the cut and
therefore encodes $I(y)=\{(1,1),(2,1),(2,2)\}$, equivalently $y=(1,2)$ and the
selected subset $\{x_2,x_4\}$.  Thick arcs of capacity $U$ enforce closure;
sample finite source/sink arcs and one dashed finite pairwise arc illustrate how
the objective is encoded.  The arc labeled $w$ denotes a generic pairwise
objective capacity derived from the threshold expansion of a Riesz interaction;
it is not itself the Riesz weight $|x_4-x_2|^{-s}$ and it is not cut in the
shown partition.}
\label{fig:tiny-cut-example}
\end{figure}
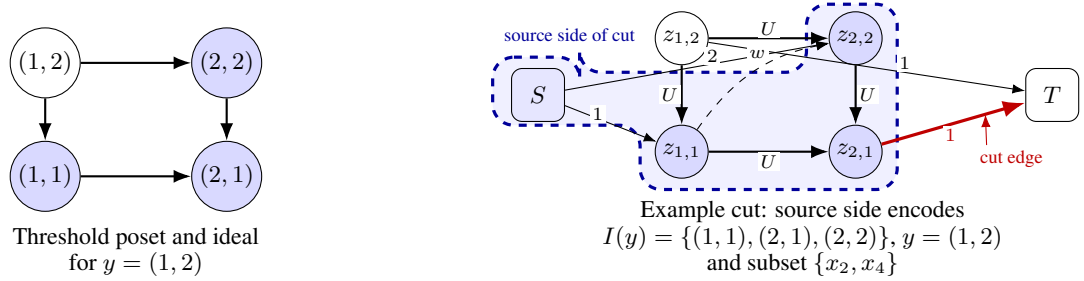

A slightly larger example can be useful because the cut then runs visibly through
an interior level of the threshold grid rather than almost around all threshold
nodes.  Take $k=3$ and $m=3$, and consider the feasible vector

\[
        y=(2,2,2).
\]
The corresponding source-side ideal is
\[
        I(y)=\{(r,q): r=1,2,3,\ q=1,2\}.
\]
Equivalently, the source side contains the first two threshold levels in each
row and excludes the third threshold level.  The selected indices are
\[
        (i_1,i_2,i_3)=(1+2,2+2,3+2)=(3,4,5),
\]
so this cut encodes the subset $\{x_3,x_4,x_5\}$.  Figure~\ref{fig:middle-cut-example}
shows this as a cut through the middle of the threshold grid.  As before, the
large $U$-arcs enforce closure and are not cut by the displayed order ideal.
The red finite arcs are illustrative objective arcs crossing from the source
side to the sink side.  The dashed red arc with capacity $w_{\mathrm{mid}}$
represents one possible pairwise objective arc obtained from the threshold
expansion; it is a graph-cut capacity associated with a mixed term and is not a
direct Riesz weight between two selected points.

\begin{figure}[htbp]
\centering
\begin{tikzpicture}[
    >=Latex,scale=1.3,
    every node/.style={font=\small},
    var/.style={draw,circle,inner sep=1.2pt,minimum size=7.5mm},
    chosen/.style={draw,circle,inner sep=1.2pt,minimum size=7.5mm,fill=blue!15},
    term/.style={font=\scriptsize,fill=white,inner sep=1pt},
    uarc/.style={->,thick,gray!65},
    cutfinite/.style={->,red!75!black,very thick},
    paircut/.style={->,red!75!black,very thick,dashed}
]
\node[draw,rectangle,rounded corners,minimum width=8mm,minimum height=8mm,fill=blue!10] (S) at (-1.7,0.95) {$S$};
\node[draw,rectangle,rounded corners,minimum width=8mm,minimum height=8mm] (T) at (6.0,0.95) {$T$};

\node[chosen] (z11) at (0,0) {$z_{1,1}$};
\node[chosen] (z12) at (0,1.35) {$z_{1,2}$};
\node[var]    (z13) at (0,2.70) {$z_{1,3}$};
\node[chosen] (z21) at (2.0,0) {$z_{2,1}$};
\node[chosen] (z22) at (2.0,1.35) {$z_{2,2}$};
\node[var]    (z23) at (2.0,2.70) {$z_{2,3}$};
\node[chosen] (z31) at (4.0,0) {$z_{3,1}$};
\node[chosen] (z32) at (4.0,1.35) {$z_{3,2}$};
\node[var]    (z33) at (4.0,2.70) {$z_{3,3}$};

\foreach \a/\b in {z12/z11,z13/z12,z22/z21,z23/z22,z32/z31,z33/z32,
                   z11/z21,z21/z31,z12/z22,z22/z32,z13/z23,z23/z33}{
  \draw[uarc] (\a) -- (\b);
}

\draw[->] (S) to[bend right=12] node[term,pos=0.42,below left=1pt and -1pt] {$1$} (z11);
\draw[->] (S) -- node[term,pos=0.52,below left=1pt and -1pt] {$1$} (z12);
\draw[cutfinite] (S) to[bend left=13] node[term,pos=0.43,above left=1pt and -2pt,text=red!75!black] {$2$} (z13);
\draw[->] (z33) -- node[term,pos=0.55,above] {$2$} (T);
\draw[cutfinite] (z31) to[bend left=6] node[term,pos=0.50,below=1pt,text=red!75!black] {$2$} (T);
\draw[cutfinite] (z32) to[bend right=10] node[term,pos=0.45,below right=-1pt and 0pt,text=red!75!black] {$1$} (T);
\draw[paircut] (z12) to[bend left=48] node[term,pos=0.32,below=6pt,text=red!75!black] {$w_{\mathrm{mid}}$} (z33);

\begin{scope}[on background layer]
    \node[draw=blue!60!black,dashed,very thick,rounded corners=10pt,
          fill=blue!6,fit=(S)(z11)(z12)(z21)(z22)(z31)(z32),inner xsep=8pt,inner ysep=8pt] (midfit) {};
\end{scope}
\node[blue!60!black,font=\scriptsize,anchor=south west]
      at ($(midfit.north west)+(0.03,0.11)$) {source side: $I(y)$};

\node[font=\scriptsize,gray!80,anchor=east] at (-0.52,0) {$q=1$};
\node[font=\scriptsize,gray!80,anchor=east] at (-0.52,1.35) {$q=2$};
\node[font=\scriptsize,gray!80,anchor=east] at (-0.52,2.70) {$q=3$};
\node[font=\scriptsize,gray!70,anchor=west] at (5.30,2.60) {gray arcs: $U$-closure};
\node[font=\scriptsize,red!75!black,anchor=west] at (5.30,2.26) {red arcs: finite arcs cut};
\node[font=\scriptsize,black,anchor=west] at (5.30,1.92) {black arcs: finite arcs not cut};
\node[align=center] at (2.1,-1.55) {Middle cut: $I(y)=\{(r,q):q\le 2\}$, $y=(2,2,2)$,\\selected subset $\{x_3,x_4,x_5\}$};
\end{tikzpicture}
\caption{A second illustrative cut for $k=3$ and $m=3$.  The blue dashed region
is the source side of the cut.  It contains all threshold nodes with
$q\leq 2$, hence encodes $y=(2,2,2)$ and the selected subset
$\{x_3,x_4,x_5\}$.  All finite (non-$U$) capacities in this illustrative graph
are shown explicitly.  The gray $U$-arcs encode closure constraints and are not
cut in the source-to-sink direction.  Red finite arcs cross the cut and
contribute to the cut value; black finite arcs do not cross the cut.  The
symbolic capacity $w_{\mathrm{mid}}$ denotes a finite pairwise graph-cut
capacity arising from the threshold expansion, not a direct Riesz interaction
weight.}
\label{fig:middle-cut-example}
\end{figure}

\subsection{How the energy is encoded by finite cut arcs}

The remaining task is to make the cut capacity equal to the Riesz objective, up
to an additive constant.  Unary threshold terms are represented by source/sink
arcs.  A positive term $c z_v$ is represented by an arc $v\to T$ of capacity
$c$, and a negative term $c z_v$ is represented as
\[
        c z_v=c+(-c)(1-z_v),
\]
which gives a constant $c$ and an arc $S\to v$ of capacity $-c$.

For a pairwise mixed term, the triangular expansion gives terms of the form
\[
        -w z_a z_b,
        \qquad w\geq 0.
\]
Use the identity
\[
        -w z_a z_b=-w z_a+w z_a(1-z_b).
\]
The first term is unary.  The second term is represented by an arc $a\to b$ of
capacity $w$, because it contributes exactly when $z_a=1$ and $z_b=0$.  Thus
each nonpositive mixed coefficient becomes one finite directed arc plus one
unary correction.  This is the standard graph-cut representation of submodular
quadratic pseudo-Boolean terms \cite{KolmogorovZabih2004}.

Putting the two ingredients together gives the algorithmic meaning of the
report's lattice reduction:
\[
\begin{array}{c}
\text{feasible index vectors} \\
\Updownarrow \\
\text{order ideals of the threshold poset} \\
\Updownarrow \\
\text{finite closed source sides of cuts} \\
\Updownarrow \\
\text{cuts whose capacities equal the Riesz energy plus a constant.}
\end{array}
\]
Therefore, one minimum cut simultaneously enforces feasibility and minimizes the
energy. Using the $z$ values forms a lattice-structured matrix from which the optimal value can be read off after applying min-cut step and then separating source-connected nodes from the others.

\subsection*{Reproducibility note}

The numerical examples and the runtime-efficiency benchmark in this note can be
reproduced with the accompanying Python scripts available at
\url{https://github.com/emmerichmtm/Riesz1DPolytime}.

\subsection*{Declaration on the use of generative AI}
The author used OpenAI tools as a programming assistant, for
language improvement, and for cross-checking.  All ideas, expository parts
of the paper, and results are by the author and the acknowledged
contributors.  Full responsibility for the content of this paper, including
any remaining errors, lies with the author.

\subsection*{Acknowledgment}

I thank the author of the MathOverflow answer cited in \cite{SD2026} for suggesting the lattice/submodularity viewpoint.

\end{document}